%% file: main.tex
\documentclass{llncs}

\usepackage{amsmath}
\usepackage{amssymb}
\usepackage{todonotes}
\usepackage{booktabs}
\usepackage[all]{xy}
\usepackage[hidelinks]{hyperref}

\usepackage{pgfplots}
\pgfplotsset{compat=1.3}

\input{macros}

\title{Implementing choreography extraction}

\author{Lu\'\i s Cruz-Filipe \and
  Fabrizio Montesi \and
  Larisa Safina}
\institute{University of Southern Denmark}
\begin{document}

\maketitle

\begin{abstract}
Choreography extraction deals with the generation of a choreography (a global description of communication behaviour)
from a set of local process behaviours.
In this work, we implement a previously proposed theory for extraction and show that, in spite of its theoretical exponential complexity,
it is usable in practice.
We discuss the data structures needed for an efficient implementation, introduce some
optimizations, and perform a systematic practical evaluation.
\end{abstract}

\section{Introduction}

Choreographies are definitions of the intended coordination that a connected system should enact,
typically given in terms of communications among concurrent processes and the composition of such
communications in larger structures~\cite{bpmn,ourTCSpaper,wscdl}.
Given a choreography, one can translate it automatically to the local behaviours of each process
that the choreography describes~\cite{BBO12,BZ07,CHY12,QZCY07}.
This translation is usually called \emph{projection}.

Projection has different uses in software development.
We give three representative examples.
In the EU project CHOReVOLUTION, choreographies are abstract specifications of communications that
are projected to skeleton implementations of processes, which the programmer can later
refine~\cite{AIT18}.
In multiparty session types, choreographies are types that are projected to local specifications of communication behaviour, which are then used for verification or monitoring purposes \cite{HYC16,NBY17}.
In choreographic programming, choreographies are high-level programs that are projected to an executable program for each process, with the guarantee that these programs will communicate and pre-/post-process the communicated data as expected \cite{M13:phd,ourTCSpaper}.

In general, programmers might need to edit the code generated by projection.
For example, they might configure, optimise, or add features to the local code of some process; or they might have to integrate legacy code written in technologies without support for choreographies. In these cases, we can support developers with the inverse operation of projection: \emph{choreography extraction}, the generation of a choreography from a set of process behaviours \cite{CMS18,ourFOSSACSpaper,LT12,LTY15}. Process behaviours are typically represented in terms of a process calculus or finite-state machines.

Extraction requires predicting how concurrent processes can communicate with each other.
Approaching this problem with brute force leads to the typical explosion of cases for static analysis of concurrent programs \cite{O18}.
For small examples, even algorithms of super-factorial worst-case complexity terminate in reasonable time \cite{LTY15}.
However, a more thorough and systematic practical evaluation of extraction is still missing. In general, the practical limits of current techniques for extraction are still largely unexplored and unclear.
The main reason is that there is still no implementation of the most efficient algorithm for extraction known to date \cite{ourFOSSACSpaper} (exponential worst-case complexity).

% The aim of this paper is to develop a tool for choreography extraction based on the theory presented in \cite{ourFOSSACSpaper}, explore its design space, and evaluate it empirically.

\paragraph{Contributions.}
We present the first implementation and systematic empirical evaluation of the extraction algorithm presented in \cite{ourFOSSACSpaper}.

Designing an implementation of \cite{ourFOSSACSpaper} required choosing adequate data structures, optimising substeps, and proving all these choices correct. All these were instrumental in achieving an implementation that successfully manages our test suite in reasonable time.
% Our implementation covers details that were left underspecified in \cite{ourFOSSACSpaper}. The most notable example is the design of an efficient decision procedure and accompanying data structures for guaranteeing the absence of livelocks, together with proving that our procedure implements the corresponding theoretical property correctly.

% The key idea of the algorithm is to execute symbolically the parallel composition of some process terms (called a \emph{network}) in the attempt of building a representative execution graph (called Symbolic Execution Graph), from which a choreography can then be computed by graph traversal.
% Exponential worst-case complexity originates from branching: if a process chooses between two alternative behaviours, we have to analyse the effects of these alternatives on the rest of the network.

We also devise a simple static analysis that partitions a network into independent sub-networks,
which allows for running our extraction algorithm in parallel over these sub-networks.
% Parallelising
% extraction leads to noticeable speed-ups when partitioning is possible, in the best cases bringing
% the exponential slowdown down to linear.\luis{I split this because it looked as though we split in conditionals.}

We investigated different heuristics for deciding which action(s) should be picked at each step of the symbolic execution of a network, and compared them to a baseline that simply picks actions at random.
% Interestingly, our experiments do not allow to conclude for the superiority of any of these strategies in general.

% Recursive behaviour makes it nontrivial to choose among the possible actions for the next step of a symbolic execution.
% Bad choices lead to backtracking, which can be computationally expensive. If a network \emph{cannot} actually be extracted, this can lead to analysing all possibilities before failing. We devised a smarter backtracking strategy that 
% propagates additional information up the call stack to limit these situations. Indeed, our testing
% reveals that this optimisation is actually crucial to get a practical implementation in the case of
% networks that cannot be extracted (e.g., because they have deadlocks, which cannot be represented in
% choreographies): without this optimisation, many failure tests did not even terminate (within more
% than 10 minutes), whereas with the optimisation they fail in time comparable to that of
% success.

Finally, we move to the first systematic empirical evaluation of extraction.
In addition to hand-crafted tests taken from previous works \cite{LTY15} (which produce the expected results), we have built an automated test suite that consists of two main parts. In the first part, we tested our extraction implementation on networks that are direct projections of choreographies. We used 1050 choreographies, randomly generated with different parameters (size, number of processes, number of procedures, and branchings).
% This part required implementing projection and a bisimilarity checker for (a minimal variation of) the calculus of core choreographies \cite{ourTCSpaper}, which are new. Bisimilarity is needed because extracted choreographies might be syntactically different from the originals (but semantically equivalent), due to differences in recursive structures.
In the second part, we devised an automatic tool that simulates the typical changes (both correct and incorrect) that are introduced when a programmer edits a local process program, and then tried to extract choreographies from the edited networks. This provides information on how quickly our program fails for unextractable networks.
We believe that our test suite is comprehensive enough to capture most situations of practical relevance, and is thus a useful reference also for the future design and implementations of new algorithms for extraction.

% Nobuko is impossible to measure. Goddamnit.

\section{Background}

We summarise the algorithm in~\cite{ourFOSSACSpaper} and the
relevant theoretical background.
For a formal treatment of our choreography calculus (including statements and proofs of the most
relevant properties) the reader is referred to~\cite{ourTCSpaper}.

\subsection{Choreographies and networks}

\paragraph{Choreographies.}
We work in a minimalistic choreography language that is a simple variant of that
in~\cite{ourTCSpaper}.
Choreographies are parameterized on sets of process names (hereafter ranged over by \pid{p},
\pid{q}, \ldots), expressions ($e$, $e'$, \ldots), labels ($\ell$, $\ell'$, \ldots) and procedure
names ($X$, $Y$, $X'$, $X_1$, \ldots).
We assume these sets to be fixed; they are immaterial for our presentation.

A choreography is a pair \tuple{C,\{X_i=C_i\}_{i\in I}}, where $I$ is a finite set of indices
and each $X_i$ is a recursion variable.
In examples, we often write choreographies as
$\rec{X_1}{C_1}{\rec{\ldots}{\ldots}{\rec{X_n}{C_n}C}}$, and call $C$ the \emph{main body} of the
choreography.
Each $C_i$ and $C$ are \emph{choreography bodies}, defined inductively as
\[C ::= \nil \mid X \mid \gencom; C \mid \gensel; C \mid \gencond\]
where: $\nil$ is the terminated choreography; $X$ is a \emph{procedure call}; $\gencom;C$ is a
choreography where $\pid p$ can evaluate expression $e$ and send the result to $\pid q$, and the
system continues as $C$; $\gensel;C$ is a choreography where $\pid p$ can send label $\ell$ to
$\pid q$, and the system continues as $C$; and $\gencond$ is a choreography where $\pid p$ evalutes
the Boolean expression $e$ and the system continues as $C_1$, if the result is \m{true}, and as
$C_2$, otherwise.

We assume choreographies to be well-formed: there are no self-communications and all procedures are defined ($C$ and each $C_i$ only use procedure names
in $\{X_i\}_{i\in I}$).
Furthermore, procedure calls are \emph{guarded}: each $C_i$ cannot be a single procedure call.
(We do allow $C$ to be a procedure call, though.)

Value communications ($\gencom$) and label selections ($\gensel$) differ in that the latter are
meant to communicate local choices (from conditions) rather than actual values.
For the purposes of this work, their distinction is immaterial.

Procedure calls can be executed by replacing them with their definition.
Furthermore, choreographies support out-of-order execution: any action can be executed as long as it
is not preceded by a conditional or another action involving any of its processes.
For example, choreography $\com{\pid p.e}{\pid q};\com{\pid r.e'}{\pid s};C$ can execute both
$\gencom$ (reducing to $\com{\pid r.e'}{\pid s};C$) or $\com{\pid r.e'}{\pid s}$ (reducing to
$\com{\pid p.e}{\pid q};C$).
A formal description of the semantics of choreographies can be found in~\cite{ourTCSpaper}.

\paragraph{Networks.}
The local counterpart to choreographies are networks, which are parallel compositions of process
behaviours.
Networks range over sets of process names, expressions, labels and procedure names as
choreographies.
A network is a map from a finite set of process names to pairs \tuple{B,\{X_i=B_i\}_{i\in I}}, where
$B$ and each $B_i$ are \emph{behaviours}, generated by the grammar
\[B ::= \nil \mid X \mid \bsend{\pid q}{e};B \mid \brecv{\pid p};B \mid \bsel{\pid q}\ell;B \mid
\bbranch{\pid p}{\ell_1:B_1,\ldots,\ell_n:B_n} \mid \cond{e}{B_1}{B_2}\]
which are the local counterparts to choreography actions: the process executing
$\bsend{\pid q}{e};B$ evaluates expression $e$, sends the result to $\pid q$, and continues as $B$;
$\bsel{\pid q}\ell;B$ is similar; their dual actions are $\brecv{\pid p};B$ (receiving a value from
$\pid p$ and continuing as $B$) and $\bbranch{\pid p}{\ell_1:B_1,\ldots,\ell_n:B_n}$ (receiving a
label from $\pid p$ and continuing as the corresponding behaviour).
The remaining syntactic categories are as in choreographies, and we call $B$ the
\emph{main behaviour} of each process.

As for choreographies, we assume some well-formedness conditions: processes do not attempt to
communicate either with themselves or to processes that are not in the network, and all procedure
calls at each process refer to procedures defined in that process.
However, we do \emph{not} require procedure calls to be guarded.
The formal semantics can also be found in~\cite{ourTCSpaper}.

\paragraph{EndPoint Projection.}
A choreography satisfying some syntactic conditions can be projected into a network by means of a
procedure called \emph{EndPoint Projection} (EPP).
The EPP of $C$ is a network containing, for each process in $C$, a behaviour obtained
by selecting the actions of $C$ that involve that process.
The EPP Theorem states that a choreography and its EPP are semantically equivalent.

\subsection{Extraction}

The extraction algorithm from~\cite{ourFOSSACSpaper} relies on symbolically executing a
network.

\begin{definition}
  \label{defn:extr-graph}
  Let $N$ be a network.
  A \emph{Symbolic Execution Graph (SEG)} for $N$ is a directed graph whose vertices are networks,
  such that:
  %\begin{itemize}
  (i)~$N$ is a node;
  (ii)~all leaves are terminated networks;
  (iii)~if a node $N_1$ has one outgoing edge, then that edge is labelled with a communication
    action $\gencom$ or $\gensel$ ending at a node $N_2$, and $N_1$ can reduce to $N_2$ by executing
    the corresponding communication from $\pid p$ to $\pid q$;
  (iv)~if a node $N_1$ has two outgoing edges, then those edges are labelled $\pid p:\m{then}$ and
    $\pid p:\m{else}$, they target $N_2$ and $N_3$, respectively, $\pid p$'s behaviour in $N_1$
    begins with a conditional over $e$, and $N_1$ reduces to $N_2$ if $e$ evaluates to $\m{true}$
    and to $N_3$ if $e$ evaluates to $\m{false}$.
%   \end{itemize}
\end{definition}
A SEG contains a possible evolution of $N$ by scheduling the order of its actions.
In order to keep the SEG finite, we only allow unfolding of procedure calls when they occur at the
head of a process involved in the reduction labeling the edge.

For systems with infinite behaviour, every SEG contains cycles (which we also call loops).
To avoid livelocks, every process must reduce inside every cycle in a SEG.
This is controled by \emph{marking} procedure calls in networks: in every network, all procedure
calls in the main behaviour of each process are annotated with $\circ$ (unmarked) or $\bullet$
(marked).
In the initial network, all calls are unmarked; if there is an edge from $N_1$ to $N_2$ and the
corresponding transition requires unfolding a procedure call, then all the procedure calls
introduced by the unfolding are marked in $N_2$.
If this would result in all procedure calls in $N_2$ to be marked, then the marking is
\emph{erased}, i.e.~all procedure calls become unmarked.
(Note that a SEG may now contain multiple occurrences of the same network, but with different
markings.)
A loop is \emph{valid} if it contains a node where every procedure call is annotated with $\circ$.
Note that a loop may still contain processes that do not reduce, as long as they have finite
behaviour (i.e., they do not contain procedure calls).
A SEG is valid if all its loops are valid.

\begin{definition}
  \label{def:extr}
  The choreography extracted from a valid SEG is defined as follows.
  We annotate each node that has more than one incoming edge with a unique procedure identifier.
  Then, for every node annotated with an identifier, say $X$, we replace each of its incoming edges
  with an edge to a new leaf node that contains a special term $X$ (so now the node annotated with
  $X$ has no incoming edges).
  This eliminates all loops in the SEG.

  We now extract a pair $\langle\mathcal D,C\rangle$ as follows, where for $C$ we start at the
  initial network and for each procedure $X$ we start from the node identified with $X$:
%   \begin{itemize}
  (i)~if the node is a leaf labeled $\nil$, we return $\nil$;
  (ii)~if the node is a leaf labeled $X$, we return $X$;
  (iii)~if the node has one transition to $N'$ labeled $\eta$, we recursively extract $C'$ from $N'$
    and return $\eta;C'$;
  (iv)~if the node has two transitions to $N_1$ and $N_2$ labeled $\pid p.e.\m{then}$ and $\pid
    p.e.\m{else}$, we recursively extract $C_1$ from $N_1$ and $C_2$ from $N_2$ and return
    $\cond{p.e}{C_1}{C_2}$.
%   \end{itemize}
\end{definition}
Note that there can be no other leaves in a valid SEG.

We illustrate these notions with an example from~\cite{ourFOSSACSpaper}.
\begin{example}
Consider the following network $N$.
\begin{align*}
  & \actor{\pid p}{}{{}\rec{X}{\bsend{\pid q}{\ast};X}{X}}
  \ \parp\ \actor{\pid q}{}{{}\rec{Y}{\brecv{\pid p};Y}{Y}} \\
  \nonumber
  \parp \ &\actor{\pid r}{}{{}\rec{Z}{\bsend{\pid s}{\ast};Z}{Z}}
  \ \parp\ \actor{\pid s}{}{{}\rec{W}{\brecv{\pid r};W}{W}}
\end{align*}

  This network has the following two SEGs:
  \[\xymatrix@R=1em{
    \actor{\pid p}{}{X^\circ}\parp\actor{\pid q}{}{Y^\circ}\parp\actor{\pid r}{}{Z^\circ}\parp\actor{\pid s}{}{W^\circ}
    \ar@/^/[d]^{\com{\pid p.\ast}{\pid q}}
    &&
    \actor{\pid p}{}{X^\circ}\parp\actor{\pid q}{}{Y^\circ}\parp\actor{\pid r}{}{Z^\circ}\parp\actor{\pid s}{}{W^\circ}
    \ar@/^/[d]^{\com{\pid r.\ast}{\pid s}}
    \\
    \actor{\pid p}{}{X^\bullet}\parp\actor{\pid q}{}{Y^\bullet}\parp\actor{\pid r}{}{Z^\circ}\parp\actor{\pid 
s}{}{W^\circ}
    \ar@/^/[u]^{\com{\pid r.\ast}{\pid s}}
    &&
    \actor{\pid p}{}{X^\circ}\parp\actor{\pid q}{}{Y^\circ}\parp\actor{\pid r}{}{Z^\bullet}\parp\actor{\pid 
s}{}{W^\bullet}
    \ar@/^/[u]^{\com{\pid p.\ast}{\pid q}}
  }\]
  Observe that the self-loops are discarded because they do not go through a node with all $\circ$ annotations.
  From these SEGs, we can extract two definitions for $X$:
  \[
  \rec{X}{\com{\pid p.\ast}{\pid q};\com{\pid r.\ast}{\pid s};X}{X}
  \qquad\mbox{ and }\qquad
  \rec{X}{\com{\pid r.\ast}{\pid s};\com{\pid p.\ast}{\pid q};X}{X}
  \]
  
  Both of these definitions correctly capture all behaviors of the network.\qed
\end{example}

\section{Implementation}

% The challenge of extraction is to build a SEG efficiently, or determining in reasonable time that
% none exists.

\subsection{Strategy}
We implemented the extraction algorithm in a depth-first manner, building the SEG as a \m{static}
object of type \code{Graph} that is accessed by different methods.
After an initial setup that creates a graph with a single node (the initial network), we call a
method \code{buildGraph} on this node.
This method builds a list of all actions that the network can execute, takes the first action in
this list, and tries to complete the SEG assuming that that is the first action executed, returning
\code{true} if this succeeds.
If this process fails, the next action in the list is considered.
If no action leads to success, \code{buildGraph} returns \code{false}.

Actions are processed by two different methods, depending on their type.
In the case of communications, \code{buildCommunication} computes the network resulting from
executing the action, and checks whether there exists a node in the graph containing it.
If this is the case, it checks whether adding an edge to that node creates a valid loop; if so, the
edge is added and the method returns \code{true}; otherwise, the method returns \code{false}.
If no such node exists, a fresh node is added with an edge to it from the current node, and
\code{buildGraph} is called on the newly created node.

The case of conditionals is more involved, since two branches need to be created successfully.
Method \code{buildConditional} starts by treating the \code{then} case, much as described above,
except that in case of success (by closing a loop or by building a new node and receiving
\code{true} from the recursive invocation of \code{buildGraph}) it does not return, but moves to the
\code{else} branch.
If this branch also succeeds, the method returns \code{true}; if it fails, then it returns
\code{false} and deletes all edges and nodes created in the \code{then} branch from the graph: this
step is essential for soundness of the method deciding loop validity (see below).

If the main call to \code{buildGraph} returns \code{true}, the graph is a SEG for the given network.
Method \code{unrollGraph} is now called to identify and split nodes corresponding to procedure calls.
Then we extract the main choreography and all procedure definitions from the relevant nodes
recursively by reading the SEG as an AST (method \code{buildChoreographyBody}).

\subsection{Recognizing bad loops}

The most critical part of extraction is deciding when a loop can be closed.
The criteria from~\cite{ourFOSSACSpaper} require all paths in the loop to include a node where the
marking is erased.
Checking this directly is extremely inefficient, as it requires retraversing a large part of the
graph; instead, we followed a strategy that does this in time at most linear on the size of the graph.

\paragraph{Marking.}
For simplicity, we mark processes instead of procedure calls.
This is a trivial change from~\cite{ourFOSSACSpaper}: initially all processes are unmarked, and
whenever a process reduces it becomes marked.
If a transition marks all unmarked processes in the network, then the marking is
erased.\footnote{See Sect.~\ref{sec:impl-services} for a later refinement of this.}
Apart from the easier bookkeeping, this approach also provides a uniform treatment of processes with
finite behaviour: since they are initially unmarked, their execution must be completed before any
loops can be entered.
This is a stronger requirement than the original, whose implications we discuss below.

\begin{lemma}
  If a loop includes a node where the marking is erased, then the loop is valid in the sense
  of~\cite{ourFOSSACSpaper}.
\end{lemma}
\begin{proof}
  We show that the loop includes a node where all processes with recursive calls are unmarked, and
  that it starts in a node where all processes with finite behaviour are deadlocked.
  The first condition is immediate; for the second, we observe that no processes with finite
  behaviour can occur (at all), since they must reduce (otherwise the marking would not be erased)
  but cannot reduce (even in many steps) to themselves (since reduction strictly reduces their size
  in the absence of procedure calls).
  \qed
\end{proof}

The converse does not hold: some networks that are extractable with the algorithm
from~\cite{ourFOSSACSpaper} become unextractable with this implementation of marking.

\begin{example}
  Consider the network
  \begin{align*}
    \actor{\pid p&}{\rec{X}{\cond{e}{\left(\bsel{\pid q}{\lleft};X\right)}{\left(\bsel{\pid q}{\lright};\bsend{\pid r}e\right)}}X} \\
    \parp\actor{\pid q&}{\rec{X}{\bbranch{\pid p}{\lleft:X, \lright:\nil}}X} \qquad
    \parp\actor{\pid r}{\brecv{\pid p}}
  \end{align*}
  This network is extractable with the conditions from~\cite{ourFOSSACSpaper} ($\pid r$ is
  deadlocked throughout the whole loop), but not with ours.
  However, note that in the extracted choregraphy $\pid r$ may be livelocked, in case $\pid p$
  always chooses the \code{then} branch of its conditional.

  This is made even more obvious if we consider the network
  \begin{align*}
    \actor{\pid p&}{\rec{X}{\cond{e}{\left(\bsel{\pid q}{\lleft};X\right)}{\left(\bsel{\pid q}{\lright};X\right)}}X} \\
    \parp\actor{\pid q&}{\rec{X}{\bbranch{\pid p}{\lleft:X, \lright:X}}X} \qquad
    \parp\actor{\pid r}{\brecv{\pid p}}
  \end{align*}
  where $\pid r$ does not appear in the extracted choreography.
\end{example}
In order to avoid accidental livelocks, we opted not to allow these networks to be extractable.
In Sect.~\ref{sec:impl-services} we show how to remove this restriction by \emph{explicitly}
allowing particular processes to be livelocked in extracted choreographies.

\paragraph{List of bad nodes.}
Rather than computing all loops that would be created when adding an edge to the graph, we instead
keep a list of \emph{bad nodes} at each node in the graph: the nodes for which adding an edge from
the current node would produce an invalid loop.
This list can be easily updated whenever a new node is created, since the graph is built depth
first: the initial node is the only one in its list of bad nodes; if $n'$ is created initially as
target of a new edge from $n$, then (a) if the transition from $n$ to $n'$ causes the marking to be
erased, the list of bad nodes at $n'$ is simply $\{n'\}$, otherwise (b) the list of bad nodes at
$n'$ is obtained by adding $n'$ to the list of bad nodes at $n$.

When attempting to add an edge between $n'$ and $n$, we first check whether the marking is erased
in the action labeling this edge; if this happens, the loop is valid.
Otherwise, we check whether $n$ is in the list of bad nodes of $n'$, in which case we reject the
loop.

\paragraph{Choice paths.}
%The intuition behind the list of bad nodes is that erasing the marking in a transition makes it ok
%to close any loop that includes that transition.
Soundness of the strategy described above relies on the fact that no new edges are added between
existing nodes that make it possible to close a loop bypassing the edge where the marking was
erased.
This is guaranteed when a communication action is selected (the corresponding node only has one
outgoing transition), but not in the case of conditionals -- and indeed problems can occur in
such a situation, as illustrated in Figure~\ref{fig:badloop-conditional}.

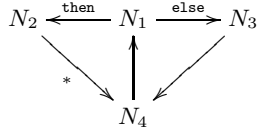
\begin{figure}
  \centering
  \[
  \xymatrix{
    N_2 \ar[dr]_{\ast} & N_1 \ar[l]_{\code{then}} \ar[r]^{\code{else}} & N_3 \ar[dl] \\
    & N_4 \ar[u]
  }
  \]

  \caption{An invalid loop that is not detected by the list of bad nodes. Extraction starts at
    $N_1$, where a conditional action is executed. In the \code{then} case, the network evolves
    first to $N_2$, then to $N_4$ through an action that erases the marking; at this stage, $N_1$ is
    not in $N_4$'s list of bad nodes, and $N_4$ can move to $N_1$ closing the loop.
    However, in the \code{else} case network $N_1$ evolves to $N_3$, which then moves to $N_4$
    without erasing the marking; however, $N_4$ is not in $N_3$'s list of bad nodes: this list was
    created from $N_1$'s, and since $N_1$ was created before $N_4$ it cannot contain $N_4$.}
  \label{fig:badloop-conditional}
\end{figure}

To avoid this problem, we restrict edge creation so that we can only add an edge to an existing node
if that node is a predecessor of the current node (i.e., it was not generated while expanding a
different branch of a conditional statement).
We do this by annotating each network with a \emph{choice path} representing the sequence of
conditional branches on which it depends.
The initial node has empty choice path; nodes generated from a communication action inherit their
parent's choice path; and nodes generated from a conditional get their parent's choice path appended
with $0$ (\code{then} branch) or $1$ (\code{else} branch).

When checking whether a node with the target network already exists in the graph, we additionally
require it to be in a node whose choice sequence is a prefix of the current node's; otherwise, we
proceed as if no such node exists.% (i.e., we create a new node).

\begin{lemma}
  \label{lem:badnodes}
  Let $n$ and $n'$ be nodes such that the choice path of $n$ is a prefix of the choice path of $n'$
  and $n'$ is the node currently being expanded.
  Adding an edge from $n'$ to $n$ creates an invalid loop iff $n$ is in the list of bad nodes of
  $n'$.
\end{lemma}
\begin{proof}
  Assume that adding an edge from $n'$ to $n$ creates an invalid loop.
  By construction, there is only one path between $n$ and $n'$: if any node between $n$ and $n'$ has
  more than one descendant, then these descendants have different choice paths, and only one of them
  can be an ancestor of $n'$.

  We show that $n$ must be in the list of bad nodes of $n'$ by induction.
  By construction, $n$ is in its own list of bad nodes.
  In the path from $n$ to $n'$, the list of bad nodes is always built extending that in the previous
  node (since the marking is not erased), so $n$ must be in the list of bad nodes of $n'$.

  Conversely, suppose that $n$ is in the list of bad nodes of $n'$.
  Then the marking has not been erased in the only path from $n$ to $n'$, thus the loop is invalid.
  \qed
\end{proof}

\subsection{Data structures}
We now summarize the data structures that we use in order to implement the functionalities described
above.

\begin{itemize}
\item \code{ConcreteNode}: the nodes created by \code{buildGraph}.
  They include a network, its choice path (a binary sequence), a unique identifier (\code{int}), a
  list of identifiers of bad nodes, and a marking (mapping process names to \code{boolean}).
\item \code{InvocationNode}: the procedure invocation nodes added when converting the SEG to a
  choreography.
  They include a procedure name and a reference to the \code{ConcreteNode} where the procedure
  starts.
\item \code{nodeHashes}: a global variable mapping each network and marking to the set of nodes with
  that network and marking.
\item \code{choicePaths}: a global variable mapping each choice path to the collection of nodes with
  exactly that choice path.
\end{itemize}

\code{ConcreteNode}s offer no surprises, as all their components have been described above.

\code{InvocationNode}s contain a reference to the procedure definition in order to make the final
step easier: method \code{unrollGraph}, which creates these nodes, returns a list of all
\code{InvocationNode}s, which is then passed on to \code{buildChoreography}.
This method iterates through this list, and constructs the different procedure definitions by
reading their name from the node and calling \code{buildChoreographyBody} on the \code{ConcreteNode}
it points to.

Global variables \code{nodeHashes} and \code{choicePaths} are included to allow finding sets of
nodes with a given property efficiently.
Both \code{buildCommunication} and \code{buildConditional} need to determine whether some network(s)
are in the graph, and \code{nodeHashes} answers this without iterating through all nodes.
Variable \code{choicePaths} is used in \code{buildConditional} in the case where construction of the
\code{then} branch succeeds but that of the \code{else} branch fails: all nodes constructed in the
\code{then} branch need to be removed from the graph, and they can be quickly retrieved via their
choice paths.

\section{Extensions and optimizations}

The next step was to implement a number of extensions and optimizations to the original algorithm.
Some of these changes aim at making the implementation more efficient in situations that we consider
may occur often enough to warrant consideration; others extend the domain of extractable
choreographies, and were motivated by practical applications.

\subsection{Parallelization}
\label{sec:impl-parallel}

In our first testing phase, we took the benchmarks from~\cite{LTY15} and wrote them as networks.
This translation was done by hand, and the goal was to ensure that the network represented the same
protocol as the communicating automata in the original work.
Of these, 3 benchmarks (\emph{alternating 2-bit}, \emph{alternating 3-bit} and \emph{TPMContract})
were not implemented.
The first uses asynchronous communication, an extension described in~\cite{ourFOSSACSpaper}
that is not implemented; the two last require non-determinism, and are not
representable in our formalism.

Several benchmarks are parallel compositions of two instances of the same network.
They exhibit a very high degree of parallelism, visibly slowing down extraction.
However, very simple static analysis can easily improve performance in such instances.
We define the network's communication graph as the undirected graph whose nodes are processes, and
where there is an edge between $\pid p$ and $\pid q$ if they ever interact.
The connected components of this graph can be extracted independently, and the choreographies
obtained composed in parallel at the end.

Theoretically, this requires adding a parallel composition constructor at the top level of a
choreography, which is straightforward.
In practice, this trivial preprocessing drastically reduces the computation time: for the (very
small) benchmarks from~\cite{LTY15}, doubling the size of the network already corresponds to a
increase in computation time of up to $35$ times, while splitting the network in two and extracting
each component in sequence keeps this factor under $2$, since the independent components can be
extracted in parallel.

We report our empirical evaluation in Table~\ref{tab:nobuko}.
These numbers are purely indicative: due to the very small size of these examples, we did not
attempt to make a very precise evaluation.
We measured the extraction time for all benchmarks, and compute the ratio between each benchmark
containing a duplicate network and the non-duplicated one.
This was done with the original, sequential, algorithm, and with the parallelized one.
All execution times were averaged over three runs.
The values themselves are not directly comparable to those from~\cite{LTY15}, since the network
implementations are substantially different, but the ratios show the advantages of our approach:
even without parallelization, our ratios are substancially lower, in line with the better
asymptotical complexity of our method shown in~\cite{ourFOSSACSpaper}.
(Note that the examples from~\cite{LTY15} where the ratio is lowest are the smallest ones, where the
execution time is dominated by the setup and command-line invocation of the different programs
used.)

\begin{table}
  \centering
  \caption{Empirical evaluation of the effect of parallelizing extraction.}
  \label{tab:nobuko}

  \begin{tabular}{l@{\hspace{1em}}rrr@{\hspace{1em}}rrr@{\hspace{1em}}rrr}
    \toprule
    Test name & \multicolumn3c{sequential} & \multicolumn3c{parallel} & \multicolumn3c{from~\cite{LTY15}} \\
    & single & double & ratio & single & double & ratio & single & double & ratio \\ \midrule
    Bargain &
    1.7 & 10.0 & \bf5.88 &
    3.0 & 3.7 & \bf1.23 &
    103 & 161 & \bf1.56 \\
    Cloud system &
    8.3 & 83.0 & \bf10.0 &
    8.6 & 8.3 & \bf0.96 &
    140 & 432 & \bf3.08 \\
    Filter collaboration &
    4.0 & 123.3 & \bf30.83 &
    5.0 & 4.7 & \bf0.93 &
    118 & 178 & \bf1.51 \\
    Health system &
    6.0 & 80.3 & \bf13.39 &
    7.3 & 11.7 & \bf1.59 &
    17 & 1702 & \bf100.12 \\
    Logistic &
    1.0 & 34.7 & \bf34.70 &
    5.3 & 16.7 & \bf3.14 &
    276 & 2155 & \bf7.81 \\
    Running example &
    7.7 & 143.3 & \bf18.61 &
    5.7 & 6.7 & \bf1.17 &
    184 & 22307 & \bf121.23 \\
    Sanitary agency &
    6.0 & 61.0 & \bf10.17 &
    8.0 & 7.3 & \bf0.92 &
    241 & 3165 & \bf13.13 \\
    %% Two-bit protocol &
    %% ? & ? & \bf? &
    %% 3.3 & 4.3 & \bf1.31 &
    %% 161 & 355 & \bf2.20 \\
    \bottomrule
  \end{tabular}
\end{table}

\subsection{Extraction strategies}
\label{sec:impl-strategy}

The performance of our implementation is depends on the choice of the network action, in cases where
there several possible options: expanding a communication generates one descendant node, but expanding
a conditional generates two descendant nodes that each need to be processed.
On the other hand, if the choreography contains cyclic behaviour, different choices of actions may
impact the size of the extracted loops (and thus also execution time).

In order to control these choices, we define \emph{execution strategies}: heuristics that guide the
choice of the next action to pick.
Strategies either take into account the syntactic type of the action (e.g.~prioritize interactions)
or the semantics of bad loops (prioritize unmarked processes), or combine them with different
priorities (prioritize unmarked processes and break ties by preferring interactions).
We also include a basic strategy that picks a random action.

All strategies are implemented in the same way: in \code{buildGraph}, the list of possible actions
that the network in the current node can execute is sorted according to the chosen criterium.

Our implementation includes the following strategies.
The abbreviations in parenthesis are used in captions of graphics.
\begin{description}
\item[\texttt{Random (R)}] Choose a random action.
\item[\texttt{LongestFirst (L)}] Prioritize the process with the longest body.
\item[\texttt{ShortestFirst (S)}] Prioritize the process with the shortest body.
\item[\texttt{InteractionsFirst (I)}] Prioritize interactions.
\item[\texttt{ConditionalsFirst (C)}] Prioritize conditionals.
\item[\texttt{UnmarkedFirst (U)}] Prioritize actions involving unmarked proceses.
\item[\texttt{UnmarkedThenInteractions (UI)}] Prioritize actions involving unmarked processes, and
  as secondary criterium prioritize interactions.
\item[\texttt{UnmarkedThenSelections (US)}] Prioritize unmarked processes, as a
  secondary criterium prioritize selections, and afterwards value communications.
\item[\texttt{UnmarkedThenConditionals (UC)}] Prioritize unmarked processes, and
  as secondary criterium prioritize conditionals.
\item[\texttt{UnmarkedThenRandom (UR)}] Prioritize unmarked processes, in random
  order.
\end{description}

We remark that \texttt{UnmarkedFirst} and \texttt{UnmarkedThenRandom} are different strategies:
\texttt{UnmarkedFirst} does not distinguish among actions involving unmarked processes, so they come
in the order of the processes involved in the network.
By contrast, \texttt{UnmarkedThenRandom} explicitly randomizes the list of possible actions, in
principle contributing towards more fairness among processes.

The experimental results in the next section show that \texttt{LongestFirst} and
\texttt{ShortestFirst} perform significantly worse than all other strategies, while \texttt{Random}
and \texttt{UnmarkedFirst} in general give the best results.

\subsection{Well-formedness check}
\label{sec:impl-wellformed}

In~\cite{ourFOSSACSpaper}, we assume networks to be well-formed.
While networks that are not well-formed are not extractable (some processes are deadlocked), this
may still take a long time to detect.
Therefore, we added an initial check that the network we are trying to extract is well-formed, and
immediately fail in case it is not.
Having this check also allows us to assume that the network is well-formed throughout the remainder
of execution.

\subsection{Livelocks}
\label{sec:impl-services}

Several examples in~\cite{LTY15} include processes that offer a service, and as such may be inactive
throughout a part (or the whole) of execution.
The algorithm in~\cite{ourFOSSACSpaper} does not allow this behaviour: when a loop is closed, every
process must either be terminated or reduce inside the loop.
In order to allow for services, we added a parameter to the extraction method containing a list of
services.
These processes are marked initially, and are not unmarked when the marking is erased.
As such, they may be unused when a loop is closed.

\subsection{Clever backtracking}

Our strategy of building the SEG in a depth-first fashion requires that, on failure, we backtrack
and explore different possible actions.
This leads to a worst-case behaviour where all possible execution paths need to be explored, in the
case that no choreography can be extracted from the original network.
However, a closer look at why a particular branch leads to deadlock allows us to avoid backtracking
in some instances: network execution is confluent, so if we reach a deadlocked state then every
possible execution reaches such a state, and extraction must fail.
It is only when extraction fails because of attempting to close an invalid loop that backtracking is
required.

To implement this refinement, we changed the return type of all methods that try to build an edge of
the SEG to a tri-valued logic.
If a method succeeds, it returns \m{ok} (corresponding to \code{true}); if it fails due to reaching
a deadlock, it returns \m{fail} (corresponding to \code{false}); and if it fails due to trying to
close an invalid loop, it returns \m{badloop}.
In recursive calls, these values are treated as follows:
% \begin{itemize}
(i) if the caller is processing a communication or the \code{else} branch of a conditional, they
  are propagated upwards;
(ii) if the caller is processing the \code{then} branch of a conditional, \m{fail} and \m{badloop}
  are propagated upwards, while \m{ok} signals that the \code{else} branch can now be treated.
% \end{itemize}

For \code{buildGraph}, a method call returning \m{ok} or \m{fail} is also propagated upwards, while
\m{badloop} signals that a different possible action should be tried.
If all possible actions return \m{badloop}, then \code{buildGraph} returns \m{fail}.
Although possibly unexpected, this is sound: due to confluence, any action that could have been
executed before that would make it possible to close a loop from this node can also be executed from
this node.

This optimization is crucial to get a practical implementation in the case of
\emph{unextractable} networks.
Most of the failure tests (Sect.~\ref{sec:testing-bad}) did not terminate before this
change, while they now fail in time comparable to that of success.

\section{Practical evaluation}

We test the performance our implementation with a three-stage plan.

In the first phase, we focused on the test cases from~\cite{LTY15}, in order to ensure that our tool
covered at least those cases.
% Since they are simple enough to examine by hand, we could also verify the results and check that
% they were correct.

In the second phase, we randomly generated 1050 choreographies and their endpoint projections by
varying four different parameters (see details below), and applied our tool to the projected networks.
% In this way, we tested whether we can extract networks that are direct projections of
% choreographies; furthermore, by checking bisimilarity between the extracted choreography and the
% original one, we obtained further soundness guarantees for our implementation.

In the third phase, we aimed at modelling the typical changes (correct or incorrect) introduced when
a programmer modifies an endpoint directly, and tried to extract choreographies from the resulting
networks.
% This yielded information about how quickly our program fails when a network is unextractable; as a
% side-result, we also got information about how often some types of protocol errors can slip through
% undetected (i.e.~the network is still extractable, but it implements a different protocol than the
% original).

We did \emph{not} attempt to generate networks directly.
Indeed, such tests are not very meaningful for two reasons: first, they do not correspond to a
realistic scenario; secondly, randomly generated networks are nearly always unextractable.
We believe our test suite to be comprehensive enough to model most situations with practical relevance.

All tests were performed on a computer running Arch Linux, kernel version 5.3.7, with an Intel Core
i7-4790K CPU and 32 GB RAM.

\subsection{Comparison with the literature}
Our first testing phase used the benchmarks from~\cite{LTY15} (see Sect.~\ref{sec:impl-parallel}).
These tests were done simply as a proof-of-concept, as their simplicity means that the measured
execution times are extremely imprecise.
%As discussed earlier, three test cases were not implementable; all others succeeded.
The results (using strategy \texttt{InteractionsFirst}) were reported in Table~\ref{tab:nobuko}.

\subsection{Reverse EPP}
In the second phase, we generated large-scale tests to check the scalability of our implementation.
Our tests consist of randomly-generated choreographies characterized by four parameters: number of
processes, total number of actions, number of those actions that should be conditionals, and a number
of procedures.

We then generate ten choreographies for each set of parameters as follows: first, we determine how
many actions and conditionals each procedure definition (including \code{main}) should have, by
uniformly partitioning the total number of actions and conditionals.
Then we generate the choreography sequentially by randomly choosing the type of the next action so
that the probability distribution of conditional actions within each procedure body is uniform.
For each action, we randomly choose the process(es) involved, again with uniform distribution, and
assign fresh values to any expression or label involved.
At the end, we randomly choose whether to end with termination or a procedure call.

This method may generate choreographies with dead code (if some procedures are never called).
Therefore there is a post-check that determines whether every procedure is reachable from
\code{main} (possibly dependent on the results of some conditional actions); if this is not the
case, the choreography is rejected, and a new one is generated.

A randomly generated choreography with conditional actions is typically unprojectable, so we amend
it (see~\cite{ourTCSpaper}) to make it projectable.
In general, this increases the size of the choreography.
Finally, we apply EPP to obtain the networks for our second test suite.

\begin{table}
  \caption{Parameters for the choreographies generated for testing.}
  \label{tab:chor-params}

  \begin{tabular}{ccrrrrr}
    \toprule 
    Test set & parameter & size & processes & \code{if}s & \code{def}s & \#\ tests \\ \midrule
    size & $k\in[1..42]$ & $50k$ & $6$ & $0$ & $0$ & $420$ \\
    processes & $k\in[1..20]$ & $500$ & 5$k$ & $0$ & $0$ & $200$ \\
    ifs (finite) & $k\in[1..4]$ & $50$ & $6$ & $10k$ & 0 & $40$ \\
    ifs (varying procedures) & $\tuple{j,k}\in[0..5]\times[0..3]$ & $200$ & $5$ & $j$ & 5$k$ & $240$ \\
%    ifs (fixed procedures) & $k\in[0..5]$ & $100$ & $50$ & $k$ & $5$ & $60$ \\
    procedures (fixed ifs) & $k\in[1..15]$ & $20$ & $5$ & $8$ & $k$ & $150$ \\ \midrule
    total &&&&&& $1050$ \\ \bottomrule
  \end{tabular}
  %1200-6-0-0
  %1450-6-0-0
  %20-5-8-11
  %20-5-8-14
\end{table}

The parameters for generation are given in Table~\ref{tab:chor-params}.
The upper bounds were determined by the limits of what we could extract with our hardware
limitations.
Four of the generated files contained tests that were too large to extract, and they were removed from
the final test set.

\paragraph{Results.}
For space reasons, we only report on the most interesting tests.
The first test shows that, predictably, for choreographies consisting of only communications, the
extraction time is nearly directly proportional\footnote{with a small overhead from needing to work
  with larger objects} to the network size, except for strategies that need to compute the size of
each process term.
We could enrich the networks with this information in order to make these strategies more efficient,
but since they perform badly in general we did not pursue this approach.

% \begin{figure}
%   \begin{tikzpicture}
%     \begin{axis}[
%         width=.5\textwidth,
%         height=15em,
%         xlabel=Number of actions,
%         ylabel=Time (msec),
%         legend pos=north west
%       ]
%       \addplot[only marks,mark=x]
%       table[x=numberOfActions,y=time(msec)-LongestFirst,col sep=tab]{stats/stats-comms-only.csv};
%       \addlegendentry{\texttt{L}}
%       \addplot[only marks,mark=square]
%       table[x=numberOfActions,y=time(msec)-ShortestFirst,col sep=tab]{stats/stats-comms-only.csv};
%       \addlegendentry{\texttt{S}}
%     \end{axis}
%   \end{tikzpicture}
%   \begin{tikzpicture}
%     \begin{axis}[
%         width=.5\textwidth,
%         height=15em,
%         xlabel=Number of actions,
%         ylabel=Time (msec),
%         legend pos=north west
%       ]
%       \addplot[only marks,mark=+]
%       table[x=numberOfActions,y=time(msec)-InteractionsFirst,col sep=tab]{stats/stats-comms-only.csv};
%       \addlegendentry{\texttt{I}}
%       \addplot[only marks,mark=triangle]
%       table[x=numberOfActions,y=time(msec)-Random,col sep=tab]{stats/stats-comms-only.csv};
%       \addlegendentry{\texttt{R}}
%     \end{axis}
%   \end{tikzpicture}
  
%   \caption{Execution time vs length, for networks consisting only of communications.
%     The omitted strategies essentially perform as \texttt{InteractionsFirst}, since there are no
%     other types of actions and no recursive procedures.}
%   \label{fig:results-test-1}
% \end{figure}

The second test is similar, but varying the number of processes (which makes for a greater number of
possible actions at each step) while keeping the size constant.
Our results show that indeed execution time grows linearly with the number of processes for
\texttt{InteractionFirst} and \texttt{Random}.
The behaviour of \texttt{LongestFirst} and \texttt{ShortestFirst} is more interesting, as the price
of computing the length of the behaviours dominates for small numbers of processes.

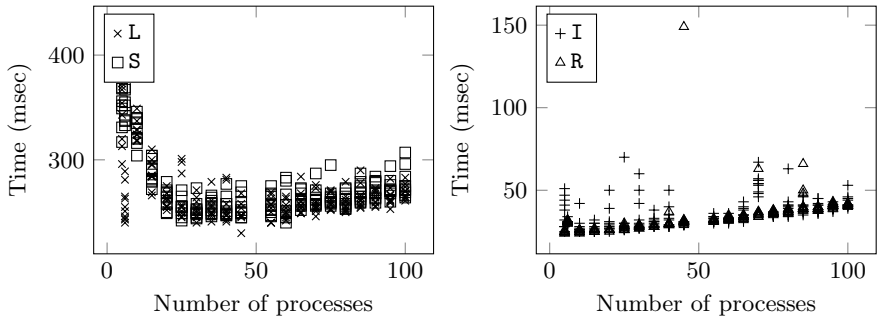
\begin{figure}
  \begin{tikzpicture}
    \begin{axis}[
        width=.5\textwidth,
        height=15em,
        xlabel=Number of processes,
        ylabel=Time (msec),
        legend pos=north west
      ]
      \addplot[only marks,mark=x]
      table[x=numberOfProcesses,y=time(msec)-LongestFirst,col sep=tab]{stats/stats-increasing-processes.csv};
      \addlegendentry{\texttt{L}}
      \addplot[only marks,mark=square]
      table[x=numberOfProcesses,y=time(msec)-ShortestFirst,col sep=tab]{stats/stats-increasing-processes.csv};
      \addlegendentry{\texttt{S}}
    \end{axis}
  \end{tikzpicture}
  \begin{tikzpicture}
    \begin{axis}[
        width=.5\textwidth,
        height=15em,
        xlabel=Number of processes,
        ylabel=Time (msec),
        legend pos=north west
      ]
      \addplot[only marks,mark=+]
      table[x=numberOfProcesses,y=time(msec)-InteractionsFirst,col sep=tab]{stats/stats-increasing-processes.csv};
      \addlegendentry{\texttt{I}}
      \addplot[only marks,mark=triangle]
      table[x=numberOfProcesses,y=time(msec)-Random,col sep=tab]{stats/stats-increasing-processes.csv};
      \addlegendentry{\texttt{R}}
    \end{axis}
  \end{tikzpicture}
  
  \caption{Execution time vs number of processes, when the total number of actions is constant.}
  \label{fig:results-test-2}
\end{figure}

The third test introduces conditionals.
Our results show that execution time varies with the \emph{total} number of conditionals in the
network, rather than with the number of conditionals in each process.
Figure~\ref{fig:results-test-3} (left) exhibits the worst-case exponential behaviour of our
algorithm, and also suggests that delaying conditionals is in general a better strategy.
Figure~\ref{fig:results-test-3} (right) shows the number of nodes created in the SEG, illustrating
that execution time is not directly proportional to this vale.

\begin{figure}
  \begin{tikzpicture}
    \begin{axis}[
        width=.5\textwidth,
        height=15em,
        xlabel=Total \#ifs in network,
        ylabel=Time (msec),
        legend pos=north west,
        scaled x ticks=false
      ]
      \addplot[only marks,mark=x]
      table[x=numberOfConditionals,y=time(msec)-LongestFirst,col sep=tab]{stats/stats-increasing-ifs-no-recursion.csv};
      \addlegendentry{\texttt{L}}
      \addplot[only marks,mark=square]
      table[x=numberOfConditionals,y=time(msec)-ConditionsFirst,col sep=tab]{stats/stats-increasing-ifs-no-recursion.csv};
      \addlegendentry{\texttt{C}}
      \addplot[only marks,mark=+]
      table[x=numberOfConditionals,y=time(msec)-InteractionsFirst,col sep=tab]{stats/stats-increasing-ifs-no-recursion.csv};
      \addlegendentry{\texttt{I}}
      \addplot[only marks,mark=triangle]
      table[x=numberOfConditionals,y=time(msec)-Random,col sep=tab]{stats/stats-increasing-ifs-no-recursion.csv};
      \addlegendentry{\texttt{R}}
    \end{axis}
  \end{tikzpicture}
  \begin{tikzpicture}
    \begin{axis}[
        width=.5\textwidth,
        height=15em,
        xlabel=Total \#ifs in network,
        ylabel=Nodes created,
        legend pos=north west,
        scaled x ticks=false
      ]
      \addplot[only marks,mark=x]
      table[x=numberOfConditionals,y=nodes-LongestFirst,col sep=tab]{stats/stats-increasing-ifs-no-recursion.csv};
      \addlegendentry{\texttt{L}}
      \addplot[only marks,mark=square]
      table[x=numberOfConditionals,y=nodes-ConditionsFirst,col sep=tab]{stats/stats-increasing-ifs-no-recursion.csv};
      \addlegendentry{\texttt{C}}
      \addplot[only marks,mark=+]
      table[x=numberOfConditionals,y=nodes-InteractionsFirst,col sep=tab]{stats/stats-increasing-ifs-no-recursion.csv};
      \addlegendentry{\texttt{I}}
      \addplot[only marks,mark=triangle]
      table[x=numberOfConditionals,y=nodes-Random,col sep=tab]{stats/stats-increasing-ifs-no-recursion.csv};
      \addlegendentry{\texttt{R}}
    \end{axis}
  \end{tikzpicture}
  
  \caption{Execution time (left) and number of nodes (right) vs total number of conditionals, for
    networks consisting only of conditionals.
    The omitted strategies essentially perform as \texttt{InteractionsFirst} or
    \texttt{ConditionalsFirst}.}
  \label{fig:results-test-3}
\end{figure}
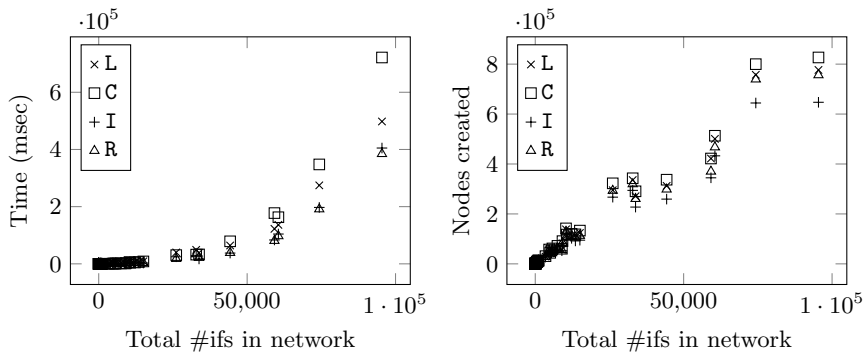

The results of the remaining tests are too complex to allow for immediate conclusions, and we omit them.

\paragraph{Correctness.}
In order to obtain confirmation of the correctness of our algorithm and its implementation, we
performed an additional verification at this point.
We implemented a naive similarity checker
that tests whether a choreography $C_1$ can simulate
another choreography $C_2$ as follows: we keep a set of pairs $R$, initially containing only the
pair \tuple{C_1,C_2}.
At each step, we choose a pair \tuple{C,C'} from $R$ and compute all actions $\alpha$ and
choreographies $C_\alpha$ such that $C$ can reach $C_\alpha$ by executing $\alpha$.
For each such action $\alpha$, we check that $C'$ can execute $\alpha$, compute the resulting
choreography $C'_\alpha$, and add the pair \tuple{C_\alpha,C'_\alpha} to $R$.
If $C'$ cannot execute $\alpha$, the checker returns \code{false}.
When all pairs in $R$ have been processed, the checker returns \code{true}.

We then check, for each test, that the original choreography and the one obtained by extraction
can simulate each other.

\begin{lemma}
  If $C_1$ and $C_2$ can simulate each other, then there is a bisimulation between $C_1$ and $C_2$.
\end{lemma}
\begin{proof}
  We first observe that the final set $R$ computed by the algorithm is always the same, regardless
  of the order in which pairs are picked.
  
  Let $R_{12}$ and $R_{21}$ be the sets built when checking that $C_2$ simulates $C_1$ and that
  $C_2$ simulates $C_1$, respectively.
  We show by induction on the construction of $R_{12}$ that $R_{12}^{-1}\subseteq R_{21}$.
  Initially this holds, since $R_{12}=\{\tuple{C_1,C_2}\}$ and \tuple{C_2,C_1} is initially in
  $R_{21}$.
  Suppose $\tuple{C,C'}\in R_{12}$ is selected for processing.
  By induction hypothesis, $\tuple{C',C}\in R_{21}$.
  For every $\alpha$ such that $C$ can execute $\alpha$ and move to $C_\alpha$, there is a unique
  choreography $C'_\alpha$ such that $C'$ can execute $\alpha$ and move to $C'_\alpha$.
  Therefore, in the step where \tuple{C',C} is selected from $R_{21}$, every such pair
  \tuple{C'_\alpha,C_\alpha} is added to $R_{21}$, hence it is in the final set.
  Thus after extending $R_{12}$ with all the pairs obtained from \tuple{C,C'} the thesis still holds.

  By reversing the roles of $C$ and $C'$, we also establish that $R_{21}^{-1}\subseteq R_{12}$.
  Therefore $R_{12}=R_{21}^{-1}$.
  It then follows straightforwardly that $R_{12}$ is a bisimulation between $C_1$ and $C_2$.
  \qed
\end{proof}

Given that bisimulation is in general undecidable and that we did not make any effort to make a
clever implementation, our program often runs out of resources without terminating.
Still, it finished in about 5\%\ of the tests (those of smaller size), always with a positive
result.

\subsection{Fuzzer and unroller}
\label{sec:testing-bad}

In the third testing phase, we changed the networks obtained by choreography projection using two
different methods.
The first method (the \emph{fuzzer}) applies transformations that are semantically wrong, and
typically result in unextractable networks (modelling programmer errors).
The second method (the \emph{unroller}) applies transformations that are semantically correct, and
result in networks that are bisimilar to the original and should be extractable (modelling
alternative implementations of the same protocol).

\paragraph{The fuzzer.}
For the fuzzer, we considered the following transformations: adding an action; removing an action; and
switching the order of two actions.
The first two always result in an unextractable network, whereas the latter may still give an
extractable network that possibly implements a different protocol.

Our fuzzer takes two parameters $d$ and $s$, randomly chooses one process in the network, deletes
$d$ actions in its definition and switches $s$ actions with the following one.
The probability distribution of deletions and swaps is uniform (all actions have the same
probability of being deleted of swapped).
We made the following conventions: deleting a conditional preserves only the the \code{then} branch;
deleting an offering preserves only the first branch offered; swapping a conditional or offering
with the next action switches it with the first action in the \code{then}/first branch; and swapping
the last action in a behaviour with the next one amounts to deleting that action.
Deleting a conditional results in an extractable network that implements a subprotocol of the
original one, while other deletions yield unextractable networks.
Swaps of communication actions may yield extractable networks, but with a different extracted
choreography; all other types of swaps break extractability.

We did not implement adding a random action, as this is covered in our tests: adding an unmatched
send from \m{p} to \m{q} can be seen as removing a receive at \m{q} from \m{p} from a choreography
that includes that additional communication.
We restricted fuzzing to one process only since in practice we can assume that endpoints are changed
one at a time.
We applied three different versions of fuzzing to all our networks: one swap; one deletion;
and two swaps and two deletions.
The results are summarized in Table~\ref{tab:fuzzing}.

\begin{table}
  \caption{Extracting fuzzed networks: for each strategy we report on the percentage of
    unextractable networks (\%) and the average time to fail in those cases (ms).}
  \label{tab:fuzzing}

  \centering
  \begin{tabular}{cc|rr|rr|rr|rr|rr|rr|rr|rr|rr|rr}
    \toprule
    && \multicolumn2{c|}{\texttt{R}} & \multicolumn2{c|}{\texttt{L}} & \multicolumn2{c|}{\texttt{S}} &
    \multicolumn2{c|}{\texttt{I}} & \multicolumn2{c|}{\texttt{C}} & \multicolumn2{c|}{\texttt{U}} &
    \multicolumn2{c|}{\texttt{UI}} & \multicolumn2{c|}{\texttt{US}} & \multicolumn2{c|}{\texttt{UC}} &
    \multicolumn2c{\texttt{UR}} \\
    $d$ & $s$ & \% & ms & \% & ms & \% & ms & \% & ms & \% & ms & \% & ms & \% & ms & \% & ms & \% & ms & \% & ms \\ \midrule
    0 & 1 & 48 & 158 & 44 & 1595 & 47 & 1667 & 47 & 203 & 46 & 283 & 48 & 219 & 46 & 218 & 47 & 218 & 45 & 258 & 45 & 215 \\
    1 & 0 & 99 & 330 & 99 & 2192 & 99 & 2177 & 99 & 297 & 99 & 525 & 99 & 306 & 99 & 300 & 99 & 301 & 99 & 338 & 99 & 260 \\
    2 & 2 & 100 & 195 & 100 & 1378 & 100 & 1581 & 100 & 158 & 100 & 181 & 100 & 124 & 100 & 120 & 100 & 115 & 100 & 131 & 100 & 133 \\ \bottomrule
  \end{tabular}
\end{table}

The differences in the percentages in the first row are due to memory running out in some cases.
It is interesting that the slowest strategies actually perform better.
In later rows most networks are unextractable; the interesting observation here is that strategies
prioritizing unmarked processes fail much faster.

\paragraph{The unroller.}
Projections of choregraphies are intuitively easy to extract because their recursive procedures are
all synchronized (they come from the same choreography).
In practice, this is not necessarily the case: programs often include ``loops-and-a-half'', where it
is up to the programmer to decide where to place the duplicate code; and sometimes procedure
definitions can be locally optimized.
For example: if $X=\com{\pid p.e}{\pid q};\com{\pid p.e'}{\pid q};\com{\pid p.e''}{\pid r};X$, then
in the extracted implementation of \pid{q} the definition of $X_{\pid q}$ can simply be
$\brecv{\pid p};X$.

Our unroller models these situations by choosing one process and randomly unfolding some procedures,
as well as shifting the closing point of some loops.
These transformations are always correct, so they should yield extractable networks, but the
extraction time may be larger and there may be higher chance for bad loops.
We generated approximately 250 tests, which we were all able to extract.
A detailed analysis of these results is left to an extended version of this paper.
% \todo[inline]{Add graph: time (original)/time (unrolled) sounds like a good idea.}

\section{Conclusions}

We successfully implemented the choreography extraction algorithm from~\cite{ourFOSSACSpaper} and
showed is practical usability by testing it on a comprehensive test suit developed for this purpose.
Our evaluation also shows that the best performing strategy when the network is extractable is
\texttt{Random}, but when unextractable networks come into play there is an advantage on using
\texttt{UnmarkedFirst}.
We believe that extraction strategies have unexplored potential, but they will need to be much more
complex in order to have a pronounced impact on performance.

\paragraph{Acknowledgements.}
The authors were supported in part by the Independent Research Fund Denmark,
Natural Sciences, grant DFF-1323-00247.

\bibliographystyle{plain}
\bibliography{biblio}

\end{document}

%% file: macros.tex
\newcommand{\m}[1]{\ensuremath{\mathsf{#1}}}
\newcommand{\pid}[1]{\m{#1}}

\newcommand{\nil}{\boldsymbol 0}

\newcommand{\code}[1]{\texttt{\upshape #1}}
\newcommand{\com}[2]{#1\;\code{-\hspace{-0.3mm}>}\;#2}
\newcommand{\gencom}{\com{\pid p.e}{\pid q}}

\newcommand{\sel}[3]{\com{#1}{#2 [#3]}}
\newcommand{\gensel}{\sel{\pid p}{\pid q}{\ell}}
\newcommand{\lleft}{\textsc{l}}
\newcommand{\lright}{\textsc{r}}

\newcommand{\cond}[3]{\m{if}\, #1 \, \m{then} \, #2 \, \m{else} \, #3}
\newcommand{\gencond}{\cond{\pid p.e}{C_1}{C_2}}
\newcommand{\rec}[3]{\m{def} \, #1  =  #2 \, \m{in} \, #3}
\newcommand{\actor}[2]{#1 \triangleright #2}
\newcommand{\bsend}[2]{{#1}!{#2}}
\newcommand{\brecv}[1]{#1?}
\newcommand{\parp}{\, \boldsymbol{|} \, }
\newcommand{\bsel}[2]{{#1}\oplus#2}
\newcommand{\bbranch}[2]{{#1}\&\{{#2}\}}

\newcommand{\tuple}[1]{\ensuremath{\langle{#1}\rangle}}